\documentclass{article}
\usepackage{graphicx} 

\usepackage{amsmath,amsfonts,stmaryrd,amssymb} 

\usepackage{enumerate} 

\usepackage[framemethod=tikz]{mdframed} 

\usepackage{listings} 
\usepackage{geometry} 

\usepackage[utf8]{inputenc} 
\usepackage[T1]{fontenc} 

\usepackage{XCharter} 

\mdfdefinestyle{commandline}{
	leftmargin=10pt,
	rightmargin=10pt,
	innerleftmargin=15pt,
	middlelinecolor=black!50!white,
	middlelinewidth=2pt,
	frametitlerule=false,
	backgroundcolor=black!5!white,
	frametitle={Command Line},
	frametitlefont={\normalfont\sffamily\color{white}\hspace{-1em}},
	frametitlebackgroundcolor=black!50!white,
	nobreak,
}

\mdfdefinestyle{file}{
	innertopmargin=1.6\baselineskip,
	innerbottommargin=0.8\baselineskip,
	topline=false, bottomline=false,
	leftline=false, rightline=false,
	leftmargin=2cm,
	rightmargin=2cm,
	singleextra={%
		\draw[fill=black!10!white](P)++(0,-1.2em)rectangle(P-|O);
		\node[anchor=north west]
		at(P-|O){\ttfamily\mdfilename};
		\def\l{3em}
		\draw(O-|P)++(-\l,0)--++(\l,\l)--(P)--(P-|O)--(O)--cycle;
		\draw(O-|P)++(-\l,0)--++(0,\l)--++(\l,0);
	},
	nobreak,
}

\mdfdefinestyle{question}{
	innertopmargin=1.2\baselineskip,
	innerbottommargin=0.8\baselineskip,
	roundcorner=5pt,
	nobreak,
	singleextra={%
		\draw(P-|O)node[xshift=1em,anchor=west,fill=white,draw,rounded corners=5pt]{%
		Question \theQuestion\questionTitle};
	},
}

\newcounter{Question} 

\mdfdefinestyle{warning}{
	topline=false, bottomline=false,
	leftline=false, rightline=false,
	nobreak,
	singleextra={%
		\draw(P-|O)++(-0.5em,0)node(tmp1){};
		\draw(P-|O)++(0.5em,0)node(tmp2){};
		\fill[black,rotate around={45:(P-|O)}](tmp1)rectangle(tmp2);
		\node at(P-|O){\color{white}\scriptsize\bf !};
		\draw[very thick](P-|O)++(0,-1em)--(O);
	}
}

\mdfdefinestyle{info}{%
	topline=false, bottomline=false,
	leftline=false, rightline=false,
	nobreak,
	singleextra={%
		\fill[black](P-|O)circle[radius=0.4em];
		\node at(P-|O){\color{white}\scriptsize\bf i};
		\draw[very thick](P-|O)++(0,-0.8em)--(O);
	}
}

\usepackage{amsmath}
\usepackage{amsthm}
\usepackage{bbm}
\usepackage{hyperref}
\usepackage{float}
\usepackage{extarrows}
\usepackage{cancel}
\usepackage{mathtools}
\usepackage{biblatex}
\usepackage[inline]{enumitem}
\usepackage{float}
\usepackage{authblk}

\usepackage{listings}
\usepackage{xcolor}
\usepackage{lmodern}
\usepackage{amsmath,array}
\usepackage{tensor}
\usepackage{mdframed}

\usepackage{easybmat}
\usepackage{feynmp-auto}
\usepackage{tikz}
\usepackage{pgfplots}
\pgfplotsset{compat=1.18}
\usepackage[most]{tcolorbox}
\usepackage{caption}

\newtcolorbox{wiparea}{
  breakable,
  colback=yellow!10,
  colframe=red!75!black,
  title=Work In Progress
}

\newtheoremstyle{plain}
  {\topsep}   
  {\topsep}   
  {\itshape}  
  {0pt}       
  {\bfseries} 
  {.}         
  {5pt plus 1pt minus 1pt} 
  {}          
\newtheoremstyle{definition}
  {\topsep}   
  {\topsep}   
  {\normalfont}  
  {0pt}       
  {\bfseries} 
  {.}         
  {5pt plus 1pt minus 1pt} 
  {}          
\newtheoremstyle{remark}
  {\topsep}   
  {\topsep}   
  {\normalfont}  
  {0pt}       
  {\itshape} 
  {.}         
  {5pt plus 1pt minus 1pt} 
  {}          

\newtheoremstyle{algorithm}
  {\topsep}   
  {\topsep}   
  {\normalfont}  
  {0pt}       
  {\bfseries} 
  {\newline\newline}         
  {5pt plus 1pt minus 1pt} 
  {}          
\newtheorem{theorem}{Theorem}

\newtheorem{claim}{Claim}
\newtheorem{primitive}{Primitive}

\theoremstyle{definition}
\newtheorem{definition}{Definition}

\theoremstyle{remark}
\newtheorem{remark}{Remark}

\theoremstyle{algorithm}
\newtheorem{algorithm}{Algorithm}

\newcommand{\x}{\mathbf{x}}
\newcommand{\y}{\mathbf{y}}
\newcommand{\z}{\mathbf{z}}
\newcommand{\bb}{\mathbf{b}}
\newcommand{\vv}{\mathbf{v}}
\newcommand{\ww}{\mathbf{w}}
\newcommand{\F}{\mathbb{F}}

\definecolor{palred}{RGB}{239, 71, 111}
\definecolor{palyellow}{RGB}{255, 209, 102}
\definecolor{palgreen}{RGB}{6, 214, 160}
\definecolor{palblue}{RGB}{17, 138, 178}
\definecolor{palblack}{RGB}{7, 59, 76}

\definecolor{palbluegreen}{RGB}{12, 176, 169}
\definecolor{palredblue}{RGB}{128, 105, 145}
\definecolor{palblueyellow}{RGB}{136, 174, 140}
\definecolor{palgreenyellow}{RGB}{131, 212, 131}
\definecolor{palredyellow}{RGB}{247, 140, 107}
\definecolor{palredgreen}{RGB}{123, 143, 136}

\usepackage[parfill]{parskip}

\makeatletter
\newcommand*{\Relbarfill@}{\arrowfill@\Relbar\Relbar\Relbar}
\newcommand*{\xeq}[2][]{\ext@arrow 0055\Relbarfill@{#1}{#2}}
\makeatother

\title{Optimization of Quadratic Constraints by Decoded Quantum Interferometry}

\author[1]{Daniel Cohen Hillel\footnote{daniel.cohen-hillel@weizmann.ac.il}}
\affil[1]{\small{\it{Department of Condensed Matter Physics, Weizmann Institute of Science, Rehovot, Israel}}}

\colorlet{ColorLightBlue}{blue!65}

\begin{document}

\newcommand{\I}[1]{1_{#1 \times #1}}
\newcommand{\SO}[1]{\text{SO}(#1)}
\newcommand{\SU}[1]{\text{SU}(#1)}
\newcommand{\GL}[1]{\text{GL}(#1)}
\newcommand{\SL}[2]{\text{SL}(#1, #2 )}
\newcommand{\gU}[1]{\text{U}(#1)}
\newcommand{\gO}[1]{\text{O}(#1)}
\newcommand{\so}[1]{\mathfrak{so}(#1)}
\newcommand{\su}[1]{\mathfrak{su}(#1)}
\newcommand{\gl}[1]{\mathfrak{gl}(#1)}
\newcommand{\asl}[1]{\mathfrak{sl}(#1)}
\newcommand{\aU}[1]{\mathfrak{u}(#1)}
\newcommand{\aO}[1]{\mathfrak{o}(#1)}
\newcommand{\tr}[1]{\text{tr\!}\left(#1\right)}
\newcommand{\set}[1]{\left\{#1\right\}}
\let\oldref\ref
\renewcommand{\ref}[1]{(\oldref{#1})}
\newcommand{\TM}{\textsuperscript{\texttt{TM}}\ }
\let\oldautoref\autoref
\renewcommand{\autoref}[1]{(\oldautoref{#1})}

\maketitle
\begin{abstract}
    A recent paper by Jordan et al. introduced Decoded Quantum Interferometry (DQI), a novel quantum algorithm that uses the quantum Fourier transform to reduce linear optimization problems -- max-XORSAT and max-LINSAT -- to decoding problems.
    In this paper, we extend DQI to optimization problems involving quadratic constraints, which we call max-QUADSAT. Leveraging a connection to quadratic Gauss sums, we give an efficient algorithm to prepare the DQI state for max-QUADSAT.\newline
    To demonstrate that our algorithm achieves a quantum advantage, we introduce the Quadratic Optimal Polynomial Intersection (quadratic-OPI) problem, a restricted variant of OPI for which, to our knowledge, the standard DQI framework offers no algorithmic speedup. We show that quadratic-OPI is an instance of max-QUADSAT and use our algorithm to optimize it. \newline
    Lastly, we present a new generalized proof of the "semicircle law" for the fraction of satisfied constraints, generalizing it to any DQI state of problems where the distribution of the number of satisfied constraints for a random assignment is sufficiently close to a binomial distribution. This condition holds exactly for the DQI state of max-LINSAT, and approximately holds in the max-QUADSAT case, with the approximation becoming exponentially better as the problem size increases. This establishes performance guarantees for our algorithm.
\end{abstract}
=====

Disclaimer (March 9, 2026): Step 7 of the algorithm contains a mistake that I currently don't know how to fix (see red annotation in \autoref{algo:max_quadsat_diag_quad_no_linear}). This invalidates \autoref{claim:efficient_algo_for_quadsat}/\autoref{algo:max_quadsat_diag_quad_no_linear} as of now. The other results of the paper still hold (mainly, the new proof of the semicircle law). I thank Noah Shutty for pointing out the mistake.

\tableofcontents


\section{Background} \label{sec:background}

In a recent paper by Jordan et al. \cite{og_dqi}, the authors present a new kind of algorithm, which they call Decoded Quantum Interferometry (DQI), to optimize the number of satisfied constraints for some constraint satisfaction problems, called \textit{max-XORSAT} and \textit{max-LINSAT}. DQI works by using the quantum
Fourier transform to reduce optimization problems to decoding problems, and was shown to provide an exponential speedup over the best currently known classical approaches for certain instances. DQI has sparked a lot of interest in the scientific community \cite{dqi_influ1, dqi_influ2, dqi_influ3, dqi_influ4}.

The max-LINSAT problem is defined as a set of $m$ constraints of the form $\left\{ f_i(\bb_i \cdot \x) : 1 \le i \le m\right\}$, with vectors $\bb_i \in \F_p^n$ and functions $f_i : \F_p \to \left\{ \pm 1\right\}$, where we refer to $+1$ as satisfied and $-1$ as unsatisfied, where the goal is to find $\x \in \F_p^n$ that satisfies the most constraints.
This problem can be seen as a maximization problem of the function $f(\x) = \sum_{i=1}^m f_i(\bb_i \cdot \x)$. We denote by $B$ the matrix whose rows are $\bb_i$.

The authors suggest preparing the \textit{DQI state}:
\begin{equation} \label{eq:dqi_state_introduction}
    \left| P(f) \right> = \sum_{\x \in \F_p^n} P(f(\x)) \left| \x \right>
\end{equation}
Where $P$ is an appropriately normalized degree-$\ell$ polynomial. One can obtain strings $\x$ that satisfy many constraints with high probability upon measuring this state with an optimally chosen polynomial. The DQI algorithm is then an efficient way to prepare the DQI state. The authors show an efficient algorithm to prepare the QFT of the DQI state when $\bb_i$ are columns of a parity-check matrix of a linear code that can efficiently correct up to $\ell$ errors.

In some cases, the DQI state can be efficiently prepared by noticing that its quantum Fourier transform has a particular structure. It can be shown that the QFT of the DQI state for max-LINSAT can be written as (normalization factors omitted)
\begin{equation}
    \big| \widetilde{P}(f) \big>
    =
    \sum_{k=0}^\ell w_k \sum_{\substack{\y \in \F_p^m\\ |\y| = k}}
    \left( \prod_{\substack{i=1 \\ y_i \ne 0}}^m \tilde{g}_i(y_i) \right) \big| B^T \y   \big>
\end{equation}
Where $w_k$ are the optimally chosen coefficients of the polynomial of degree $\ell$ in \autoref{eq:dqi_state_introduction}, and $\tilde{g}_i$ are the Fourier transforms of the shifted and rescaled constraints functions of $f_i$, for precise definition see \autoref{def:normalized_constraints}. The core of the DQI algorithm is the following steps:
\begin{enumerate}
    \item Prepare the Dicke state \cite{dicke1954coherence}
        \begin{equation}
            \sum_{\substack{\y \in \F_p^m\\ |\y| = k}} \big| \y \big>
        \end{equation}
    \item Reversibly compute the matrix multiplication $B^T \y$ into an ancilla
        \begin{equation}
            \sum_{\substack{\y \in \F_p^m\\ |\y| = k}} \big| \y \big>\big| B^T \y \big>
        \end{equation}
    \item Uncompute the $\left| \y \right>$ register using the $\left| B^T \y \right>$ register, which is exactly the decoding problem of $k$ bit-flip errors, where $B^T\y$ is the syndrome.
        \begin{equation}
            \sum_{\substack{\y \in \F_p^m\\ |\y| = k}} \big| B^T \y \big>
        \end{equation}
\end{enumerate}
Decoding is generally an NP-hard problem, but if we had chosen a max-LINSAT problem such that $B^T$ becomes the parity-check matrix of a linear code with an efficient decoding scheme, then this step could be done efficiently. This was shown, for example, that for the Optimal Polynomial Intersection (OPI) problem, where the $B^T$ matrix becomes exactly the parity-check matrix of a Reed-Solomon code\cite{reedsolomon1960polynomial}, which has efficient decoders up to half the distance of the code, such as the Berlekamp-Massey algorithm \cite{berlekamp2015algebraic}.

For the analysis of the algorithm, the authors find an analytical expression for the expected number of satisfied constraints $s$ upon measuring the DQI state. The expected fraction of satisfied constraints, $\left< s \right> / m$, is given by
\begin{equation} \label{eq:semicircle_introduction}
    \frac{\left< s \right>}{m}
    =
    \left(
    \sqrt{\frac{\ell}{m} \left(1 - \frac{r}{p}\right)}
    +
    \sqrt{\frac{r}{p} \left(1 - \frac{\ell}{m}\right)}
    \right)^2
\end{equation}
For $\frac{r}{p} \le 1 - \frac{\left< s \right>}{m}$ and $\frac{\left< s \right>}{m} = 1$ otherwise. Where $m$ is the number of constraints, $\ell$ is the degree of the polynomial, $p$ is the prime modulus of the field we work in, and $r$ is the number of satisfying assignments for each $f_i$ (which is assumed to be equal for all $i$). This expression is informally referred to as the "semicircle law". This equation was derived by defining the "number of satisfied constraints" observable $S_f$, looking at its expectation value in the DQI state, then finding the maximal value for the optimal polynomial.



\section{Results} \label{sec:results}

In this paper, we first consider what happens when one replaces the linear constraints $f_i (\bb_i \cdot \x)$ of max-LINSAT with quadratic form constraints $f_i(\x^T C_i \x)$ for some symmetric matrices $C_i$. We call this problem \textit{max-QUADSAT}.
\begin{definition}[max-QUADSAT]\label{def:max_quadsat}
    Let $\F$ be a finite field. For $i = 1, \ldots, m$, let $f_i : \F \to \set{+1, -1}$ be arbitrary functions.
    Given the vectors $\mathbf{b}_1,\ldots,\mathbf{b}_m \in \F^n$ and the symmetric matrices $C_1, \ldots C_m \in \F^{n \times n}$, the max-QUADSAT problem is to find $\x \in \F^n$ maximizing the objective $
        f(\x) = \sum_{i=1}^m f_i\left(
            \mathbf{b}_i^T \x + \x^T C_i \x
        \right)
    $
\end{definition}
This manuscript will assume the same assumptions made in \cite{og_dqi}. That is, $m = \text{poly}(n)$, and $\F = \F_p$ where $p$ is a prime whose size is polynomial in $n$. Therefore, $f_1, \ldots f_m$ can be specified explicitly. We take $f_i = +1$ to mean that the constraint $i$ is satisfied and $f_i = -1$ to mean the constraint is unsatisfied. Also, we will sometimes denote $B$ as a matrix whose rows are $b_i$, and $C$ as a vector of the matrices $C_i$.
\begin{remark}\label{remark:equal_preimage_assumption}
    We restrict ourselves to situations where the preimage of the constraints $f_i$ has the same cardinality. That is, let $F_i = f_i^{-1}(+1)$ be the preimage set of $+1$ for the $i^{\text{th}}$ constraints, then for all $i = 1,\dots , m$, $r = \left| F_i \right| \in \set{1, \dots, p-1}$, where $r$ does not depend on $i$.
\end{remark}

Our goal is to prepare the DQI state for the objective function of max-QUADSAT, which is done in \autoref{sec:dqi_for_max_quadsat}. The DQI state for max-QUADSAT is the same as the one for max-LINSAT, with a different objective function.
\begin{definition}[DQI state for max-QUADSAT]\label{def:dqi_state_max_quadsat}
    Let $P(f)$ be a degree $\ell$ polynomial of the objective as defined in \autoref{def:max_quadsat}. Then the DQI state for max-QUADSAT is
    \begin{equation} \label{eq:dqi_state}
        \left| P(f) \right> = \sum_{\x \in \F_p^n} P(f(\x)) \left| \x \right>
    \end{equation}
\end{definition}
We can regard $P(f)$ as a multivariate polynomial in the constraints $f_1,\cdots, f_m$. The first significant result is that, in some cases, there's an efficient algorithm to prepare the DQI state for max-QUADSAT.
\begin{theorem}[efficient algorithm for max-QUADSAT DQI] \label{claim:efficient_algo_for_quadsat}
    There's an efficient algorithm, \autoref{algo:max_quadsat_diag_quad_no_linear}, to prepare the DQI state for max-QUADSAT \autoref{def:max_quadsat} when $\bb_i=0$ and the matrices $C_i$ are diagonal, and the diagonals of the matrices $C_i$ are columns of a parity-check matrix that has an efficient decoding scheme.
\end{theorem}
The main difference of this algorithm compared to the algorithm for preparing DQI state of max-LINSAT, is that for max-LINSAT, the information about the functionals $\bb_i$ is encoded in which strings appear in the superposition of the QFT of the DQI state, and for max-QUADSAT, the information about $C_i$ is encoded in the phases. This happens because when we go from a linear function in $x$ to a quadratic $x^2$, a sum over all roots of unity becomes a quadratic Gauss sum \cite{murty2017evaluation}, which has a constant amplitude but varying phase. Full explanation is given in \autoref{sec:dqi_quantum_state_for_max_quadsat}.

We can illustrate the difference between max-LINSAT and max-QUADSAT algorithms by looking at a simple example, where $n=m=\ell =1$, then the DQI state is just $\sum_{x \in \F_p} f_1(c_1 x^2) \left| x \right>$, the QFT of this state is $\sum_{x,y \in \F_p} \omega_p^{x \cdot y} f_1(c_1 x^2) \left| y \right>$, and if we write $f_1$ in terms of its Fourier transform $\tilde{f}_1$, we get
\[
  \sum_{x, y \in \F_p} \omega_p^{x \cdot y} \sum_{z \in \F_p } \omega_p^{-z \cdot c_1 x^2} \tilde{f}_1(z) \left| y \right>
\]
And if we reorder the sums, we get
\[
  \sum_{y,z \in \F_p} \tilde{f}_1(z) \left(\sum_{x \in \F_p}\omega_p^{x \cdot y-z \cdot c_1 x^2}\right) \left| y \right>
\]
The inner sum over $x$ is a quadratic Gauss sum, hence it is a complex value whose magnitude is always $\sqrt{p}$ and its phase contains information about $c_1$, $y$, and $z$. This differs from max-LINSAT, where the inner sum would become a Kronecker delta that would pick out specific values in the superposition. The quantum algorithm is then an efficient algorithm for preparing this state, taking the QFT, and measuring.

This state is prepared similarly to the max-LINSAT case, where one starts from the Dicke state, reversibly computes the matrix multiplication into an ancilla, and decodes it to uncompute the first register. You can use this register to prepare the quadratic phase as shown in \autoref{prim:quadratic_phase}.

Next, in \autoref{sec:problems}, we present the Quadratic Optimal Polynomial Intersection (quadratic-OPI) problem, a variant of OPI that can be written as an instance of max-QUADSAT. The OPI problem, as defined in \cite{og_dqi}, is the following:
\begin{definition}[Optimal Polynomial Intersection] \label{def:opi}
     Given integers $n < p-1$ with $p$ prime, an instance of the Optimal Polynomial Intersection (OPI) problem is as follows. Let $F_1 , \dots , F_{p-1}$ be subsets of the finite field $\F_p$. Find the degree $n - 1$ polynomial in $Q \in \F_p [y]$ that maximizes $f_\text{OPI}(Q) = \left|\set{y \in \set{1, . . . , p - 1} : Q(y) \in \F_y }\right|$, i.e. intersects as many of these subsets as possible.
\end{definition}
Our variant, quadratic-OPI, is almost the same problem as OPI, with the additional requirement that the polynomial coefficients are quadratic residues.
\begin{definition}[Quadratic OPI] \label{def:quadratic_opi}
    Quadratic OPI is an instance of OPI \autoref{def:opi}, with the additional requirement that the coefficients of the optimal polynomial must be quadratic residues.
\end{definition}
This problem is at least as hard as regular OPI, since solutions to quadratic-OPI are solutions to regular OPI, and we achieve a quantum speedup for it using our modified version of DQI. 

\begin{theorem}\label{theorem:quadratic_opi_is_max_quadsat}
    Quadratic OPI is a special case of max-QUADSAT, and the corresponding DQI state can be efficiently prepared by \autoref{algo:max_quadsat_diag_quad_no_linear} with $\ell = \lfloor \frac{n+1}{2}\rfloor$.
\end{theorem}

The last significant result is that the semicircle law \autoref{eq:semicircle_introduction} is more general than was shown in \cite{og_dqi}, and can be derived differently. This is important for us because the original derivation of the semicircle law is highly dependent on the linear structure of the DQI state, which does not hold for our quadratic case. Our derivation only depends on the probability distribution of a random string to satisfy $s$ constraints of max-LINSAT/QUADSAT. We start by re-deriving the following lemma:


\begin{theorem} (re-derivation of \cite[lemma 9.2]{og_dqi})\label{claim:rederive_lemma92}
    Let $f(\x) = \sum_{i=1}^m f_i(B_{ij}x_j)$ be a max-LINSAT objective function with matrix $B \in \F_p^{m \times n}$ for a prime $p$ and positive integers $m$ and $n$, such that $m>n$. Suppose that $\left|f_i^{-1}(+1) \right| = r$ for some $r \in \set{1, \dots, p-1}$. Let $\big< s^{(m, \ell)} \big>$ be the expected number of satisfied constraints for the symbol string obtained upon measuring the DQI state \autoref{def:dqi_state_max_quadsat}. If $2\ell + 1 < d^{\perp}$ where $d^\perp$ is the minimum distance of the code $C^\perp = \set{\mathbf{d} \in \F_p^m : B^T \mathbf{d} = \mathbf{0}}$, then
    \begin{equation} \label{eq:lemma92_result}
        \big< s^{(m, \ell)} \big>
        =
        \frac{mr}{p}
        +
        \frac{\sqrt{r (p - r)}}{p} \ww^\dagger A^{(m, \ell, d)} \ww
    \end{equation}
    where $\ww = (w_0, \dots, w_\ell)^T$ and $A^{(m, \ell, d)}$ is the $(\ell+1) \times (\ell+1)$ symmetric tridiagonal matrix
    \[
        A^{(m, \ell, d)}
        =
        \begin{pmatrix}
            0 & a_{1} \\
            a_{1} & d & a_{2} \\
            & a_{2} & 2d & \ddots \\
            & & \ddots  & & a_{\ell}\\
            &&& a_{\ell} & \ell d
        \end{pmatrix}
    \]
    with $a_k = \sqrt{k(m-k+1)}$ and $d=\frac{p-2r}{\sqrt{r(p-r)}}$
\end{theorem}

The semicircle law arises by finding the maximal eigenvalue of $A^{(m, \ell, d)}$. Using it to get the maximal value of $\ww^\dagger A^{(m, \ell, d)} \ww$ when we are constrained by normalization to $\ww^\dagger \ww = 1$, this is precisely the contents of Theorem 4.1 in \cite{og_dqi}. A similar matrix problem was also encountered in \cite{barg2006spectral} from the same context of the recurrence relation of the Krawtchouk polynomials. We also note that related results were found in \cite{marwaha2025complexitydecodedquantuminterferometry} after we submitted this paper.

The following is a sketch of the proof of \autoref{claim:rederive_lemma92}

\begin{enumerate}
    \item Show that the probability of measuring a string $\x \in \F_p^n$ that satisfies $s$ constraints in the DQI state depends only on the probability distribution $N(s)$ of satisfying $s$ constraints for a random string. This probability distribution is $\text{Pr}_{\text{DQI}}[s] = |P^*(s)|^2 N(s)$, where $P^*$ is related to the DQI polynomial $P(f)$ by shifting and rescaling. This fact arises immediately by looking at the $m$ subspaces of strings that satisfy the same number of constraints.
    \item Show that for max-LINSAT, $N(s)$ is very close to a binomial distribution. It is so close that the expectation value of $s$ in the DQI state remains unchanged if $N(s)$ is replaced by a binomial distribution, $N_{\text{binom}}(s)$, where the binomial distribution arises by treating each constraint as an independent variable that has probability $r/p$ of being satisfied. This is true because $B^T$ is a parity-check matrix of a linear code that can correct up to $\ell$ errors, so every set of $2\ell+1$ columns in $B^T$ is independent, hence any subset of $2\ell+1$ constraints is independent. Because of this, the first $2\ell + 1$ moments of $N(s)$ and $N_{\text{binom}}(s)$ match, which is enough to show the expectation is the same for the distributions $|P(s)|^2 N(s)$ and $|P(s)|^2 N_{\text{binom}}(s)$.
    \item Turn the problem of finding a polynomial that maximizes the expectation $\left< s \right>_{\text{DQI}} \propto \sum_{s=0}^m s |P(s)|^2 N(s)$ into a linear algebra problem by defining an inner product $\left< p, q \right>_{\text{binom}}:= \sum_{s=0}^m p(s) q(s) N_{\text{binom}}(s)$, then $\left< s \right>_{\text{DQI}} \propto \left< sP, P \right>$. Now we can use the Krawtchouk polynomials \cite{krawtchouk_koekoek1996askeyschemehypergeometricorthogonalpolynomials}, which are a set of polynomials that are orthogonal with respect to this inner product. Using the three-term recurrence relation of the Krawtchouk polynomials (guaranteed by Favard's theorem \cite{chihara2011introduction}), we get \autoref{eq:lemma92_result}.
\end{enumerate}


All of this was for max-LINSAT, but because the only property we needed is that any set of $2\ell + 1$ constraints is independent, the proof can also work for max-QUADSAT.
\begin{claim}[semicircle holds for max-QUADSAT] \label{claim:quadsat_also_binomial}
    The result of \autoref{claim:rederive_lemma92}, equation \autoref{eq:lemma92_result}, is also true for the DQI state of max-QUADSAT when $\bb_i = 0$ and the $C_i$ matrices are diagonal, in the limit $\text{rank}(C_i) \to \infty$ for all $i$.
\end{claim}
The proof of this is almost identical to the new proof of \autoref{claim:rederive_lemma92}. The only caveat is that, unlike max-LINSAT, if $\x \in \F_p^n$ is uniformly random in $\F_p^n$, then $\x^T C_i \x$ is \textit{not} uniformly random in $\F_p$, so we can't claim every constraint independently has probability $r/p$ to be satisfied. Fortunately, thanks to the well-known square root cancellation property of character sums, $\x^T C_i \x$ becomes exponentially close to being uniformly random in $\F_p$ as the rank of $C_i$ increases. The exact claim is:
\begin{claim}[quadratic form is close to uniform] \label{claim:quadsat_is_almost_uniform_intorduction}
    Let $\x \in \F_p^n$ be uniformly distributed, and let $C \in \F_p^{n \times n}$ be a diagonal matrix of rank $r$. Then for all $a \in \F_p$, the probability $\text{Pr}\left(\x^T C \x = a\right)$ is exponentially close to a uniform distribution, that is
    \begin{equation} \label{eq:quadsat_is_almost_uniform_introduction}
        \left|\text{Pr}\left(\x^T C \x = a\right) - \frac{1}{p}\right| \le p^{-r/2}
    \end{equation}
\end{claim}
This is true because of square root cancellation, and proven exactly in \autoref{sec:opti_app_to_max_quadsat}.

\section{Discussion and Open Questions} \label{sec:discussion}

In this paper, we presented an algorithm that extends DQI to problems of quadratic constraints. We then demonstrated its power by applying it to the quadratic-OPI problem. Lastly, we found a new proof for the semicircle law, which also generalizes the conditions under which this law holds. Specifically, it holds for the DQI state in max-QUADSAT. This demonstrates that our algorithm provides a quantum advantage for the quadratic-OPI problem.



There are several promising directions for future work. First, it would be natural to consider constraint families that combine linear and quadratic terms at the same time (e.g., functionals would be of the form $\bb_i \cdot \x + \x^T C_i \x$). Second, we restricted attention to quadratic forms whose matrices are diagonal. Relaxing this diagonal assumption would substantially enlarge the problems that can be represented (because we could describe constraints on products $x_ix_j$), and it would also allow us to apply DQI to more interesting codes (because the decoding problem becomes decoding of a syndrome matrix, as opposed to a syndrome vector). We conjecture that fully removing the diagonality restriction in the general case is not achievable using the exact same techniques of this paper, but there could be interesting families of matrices for which our methods still work.

It may also be interesting to explore what new problems can be described with max-QUADSAT. We note that every instance of max-LINSAT has a corresponding quadratic version that has the same performance guarantees of DQI, similarly to how quadratic-OPI was defined.




\textbf{Acknowledgments:} I thank Dorit Aharonov for introducing me to this problem and for her guidance.

\section{DQI for max-QUADSAT} \label{sec:dqi_for_max_quadsat}

\subsection{DQI Quantum State for max-QUADSAT} \label{sec:dqi_quantum_state_for_max_quadsat}
Here, we will see what happens to the DQI state when considering a quadratic functional $\x^T C_i \x$ for the constraint, as opposed to $\bb_i \cdot \x$ that was used for max-LINSAT.

For reasons that would become apparent later, it would be convenient to work with "shifted and rescaled" constraints, defined as
\begin{definition}[shifted-rescaled constraints]\label{def:normalized_constraints}
    Let $f_i$ be the constraint $i$ in max-QUADSAT \autoref{def:max_quadsat}. Then we define
    \[
        g_i(x) := \frac{f_i(x) - \overline{f}}{\varphi}
    \]
    where $\overline{f} = \frac{1}{p}\sum_{x \in \F_p} f_i(x)$ is the mean of $f_i$, and $\varphi = \sqrt{\sum_{y \in \F_p}\left|f_i(y) - \overline{f}_i\right|^2}$ is the standard deviation of $f_i$.
\end{definition}
\begin{remark}
The reason that the mean and standard deviation of $f_i$ were the same for all $i=1, \dots, m$ in \autoref{def:normalized_constraints} is that we assumed that the preimage of $\pm 1$ had the same cardinality; see \autoref{remark:equal_preimage_assumption}. \footnote{It follows that $\overline{f} = \frac{(+1)\cdot r + (-1) \cdot (p - r)}{p} = 2 \frac{r}{p} - 1$ and
    $\varphi 
    = \sqrt{r|1 - (2 \frac{r}{p} - 1) |^2 + (p-r)|-1 - (2 \frac{r}{p} - 1) |^2} 
    = 2\sqrt{r(1-\frac{r}{p})}$}
\end{remark}

From \autoref{def:normalized_constraints} it immediately follows that the Fourier transform of $g_i$,
\begin{equation} \label{eq:g_fourier}
    \tilde{g}_i(y) = \frac{1}{\sqrt{p}} \sum_{x \in \F_p} \omega_p^{xy} g_i(x)
\end{equation}
vanishes at $y=0$ and is normalized (i.e., $\sum_{x \in \F_p} \left| g_i(x) \right|^2 = \sum_{y \in \F_p} \left| \tilde{g}_i(y) \right|^2 = 1$).

The new "objective function" $g(\x)$ is defined straightforwardly as $g(\x) = \sum_{i=1}^m g_i(b_i^T \x + \x^T C_i \x)$, and is related to the original $f(\x)$ by 
\begin{equation}\label{eq:g_to_f_relation}
    f(\x) = g(\x) \varphi + m \overline{f}
\end{equation}

\begin{definition}[elementary symmetric polynomials]
    Let $k=0,1,2,\dots$ and $m=1,2,\dots$. The degree-$k$ elementary symmetric polynomial of $m$-variables is defined as
    \begin{equation}\label{eq:elementary_symmetric_polynomial}
        P^{(k)}(f_1, \dots, f_m) = \sum_{\substack{i_1, \dots , i_k \\ \text{distinct}} } \prod_{j=1}^k f_{i_j}
    \end{equation}
\end{definition}

\begin{claim}[DQI polynomial can be written as sum of elementary symmetric polynomials] \label{claim:dqi_polynomial_can_be_written_as_sum_of_symmetric_polys}
    The polynomial from \autoref{def:dqi_state_max_quadsat} can be written as a linear combination of the elementary symmetric polynomials
    \[
        P(f) = \sum_{k=0}^\ell u_k P^{(k)}(f_1,\ldots,f_m)
    \]
    Where $P^{(k)}(f_1,\ldots,f_m)$ is the degree-$k$ elementary symmetric polynomial, and $u_k$ is a scalar. Moreover, we can transform the polynomial $P(f(\x))$ into an equivalent polynomial $Q(g(\x))$ of the same degree by using \autoref{eq:g_to_f_relation} and absorbing $\varphi$ and $m\overline{f}$ into the coefficients.
\end{claim}
\begin{proof}
    See Appendix A of \cite{og_dqi}.
\end{proof}
\begin{definition}[elementary symmetric polynomial state] \label{def:elementary_symmetric_polynomial_state}
    Let $P^{(k)}(f_1, \dots, f_m)$ be the degree-$k$ symmetric elementary polynomial with $m$ variables, then we define the state
    \begin{equation}\label{eq:elementary_symmetric_polynomial_state} 
        \big| P^{(k)} \big\rangle = 
        \frac{1}{\sqrt{p^{n-k} \binom{m}{k} }}
        \sum_{\x \in \F_p^n} P^{(k)}\left(g_1(b_1^T\x + \x^T C_1 \x), \dots, g_m(b_m^T\x + \x^T C_m \x)\right) \left| \x \right> 
    \end{equation}
\end{definition}
It follows from \autoref{claim:dqi_polynomial_can_be_written_as_sum_of_symmetric_polys} and \autoref{def:elementary_symmetric_polynomial_state} that the DQI state \autoref{def:dqi_state_max_quadsat} can be written as a superposition of elementary symmetric polynomial states
\begin{equation} \label{eq:dqi_state_as_sum_of_sym_elem_poly}
    \left| P(f) \right> = \sum_{k = 0}^\ell w_k \big| P^{(k)} \big\rangle
\end{equation}
Thus, we will shift our focus from now on to preparing the elementary symmetric polynomial states.

\begin{claim}[elementary symmetric polynomial states for max-QUADSAT]
    The elementary symmetric polynomial state \autoref{def:elementary_symmetric_polynomial_state} can be written as
    \begin{equation} \label{eq:elem_sym_poly_state_rewritten_hamming}
        \big| P^{(k)} \big>
        =
        \frac{1}{\sqrt{p^n \binom{m}{k}}}
        \sum_{\x \in \F_p^n}
        \sum_{\substack{y \in \F_p^m \\ | y | = k}}
        \omega_p^{-(B^T y)^T\x - \x^T (C \cdot y) \x }
        \prod_{\substack{i=1\\y_i\ne 0}}^m
        \tilde{g}_{i}(y)
        \left| \x \right>
    \end{equation}
\end{claim}

\begin{proof}
By definition \autoref{def:elementary_symmetric_polynomial_state}
\begin{equation} \label{eq:elem_sym_poly_of_g}
    P^{(k)}\left(g_1(b_1^T\x + \x^T C_1 \x), \dots, g_m(b_m^T\x + \x^T C_m \x)\right)
    =
    \sum_{\substack{i_1, \dots , i_k \\ \text{distinct}} } \prod_{j=1}^k g_{i_j}(b_{i_j}^T\x + \x^T C_{i_j} \x)
\end{equation}
Substituting the Fourier transform of $g_{i_j}$, \autoref{eq:g_fourier}\footnote{Note that \autoref{eq:g_fourier} is the Fourier transform, $\tilde{g}$ in terms of $g$, but we plug in the inverse Fourier transform, $g$ in terms of $\tilde{g}$} into \autoref{eq:elem_sym_poly_of_g} yields
\begin{equation} \label{eq:elem_sym_poly_of_g_fourier}
    P^{(k)}\left(g_1(b_1^T\x + \x^T C_1 \x), \dots, g_m(b_m^T\x + \x^T C_m \x)\right)
    =
    \sum_{\substack{i_1, \dots , i_k \\ \text{distinct}} } \prod_{j=1}^k \sum_{y_{i_j} \in \F_p}
    \frac{1}{\sqrt{p}}
    \omega_p^{-\left( b_{i_j}^T\x + \x^T C_{i_j} \x \right) \cdot y_{i_j}}\tilde{g}_{i_j}(y_{i_j})
\end{equation}

The sum $\sum_{\substack{i_1, \dots , i_k \\ \text{distinct}} } \prod_{j=1}^k \sum_{y_{i_j} \in \F_p}$ can be rewritten in terms of $\sum_{\substack{\y \in \F_p^m \\ | \y | = k}}$.
We can thus rewrite \autoref{eq:elem_sym_poly_of_g_fourier} as 
\begin{equation} \label{eq:elem_sym_poly_of_g_fourier_rewrite_hamming}
    P^{(k)}\left(g_1(b_1^T\x + \x^T C_1 \x), \dots, g_m(b_m^T\x + \x^T C_m \x)\right)
    =
    \frac{1}{\sqrt{p^k}}
    \sum_{\substack{\y \in \F_p^m \\ | \y | = k}}
    \omega_p^{-(B^T \y)^T\x - \x^T (C \cdot \y) \x }
    \prod_{\substack{i=1\\y_i\ne 0}}^m
    \tilde{g}_{i}(y_i)
\end{equation}
Substituting that back into \autoref{def:elementary_symmetric_polynomial_state} yields the desired result
\end{proof}

Let's define an (oddly specific) operator to help us with some notation later 
\begin{definition}[conditional quadratic phase operator] \label{def:F_alpha_very_convinient}
    We denote the operator $F_\alpha$ acting on a quantum register in $\F_p$ as
    \begin{equation} \label{eq:F_alpha_very_convinient}
        F_\alpha \left| x \right>
        =
        \begin{cases}
            p \left| \alpha \right> & \text{if $x = 0$} \\
            \chi(x) g(1;p) \sum_{z \in \F_p} \omega_p^{-x^{-1} (z-\alpha)^2 / 4} \left| z \right> & \text{otherwise}
        \end{cases}
    \end{equation}
    To simplify notation, we will sometimes write
    \[
       F_\alpha \left| x \right> = p \delta_{x,0} \left| \alpha \right> + \chi(x) g(1;p) \sum_{z \in \F_p} \omega_p^{-x^{-1} (z-\alpha)^2 / 4} \left| z \right>
    \]
\end{definition}
And now we can look at the Fourier transform of the quadratic DQI state. From now on, we will always assume the matrices $C_i$ are diagonal.
\begin{claim}[QFT of elementary symmetric polynomial] \label{claim:qft_of_elem_sym_poly_state_simplified}
    The QFT of \autoref{eq:elem_sym_poly_state_rewritten_hamming} can be written as
    \begin{equation} \label{eq:qft_of_elem_sym_poly_state_simplified}
    \big| \tilde{P}^{(k)} \big>
    =
    \sum_{\substack{\y \in \F_p^m \\ | \y | = k}}
    \prod_{\substack{i=1\\y_i\ne 0}}^m
    \tilde{g}_{i}(y_i)
    \prod_{j=1}^n
    F_{(B^T\y)_j} \left| (D^T \y)_j \right>
\end{equation}
\end{claim}
\begin{proof}
The QFT of \autoref{eq:elem_sym_poly_state_rewritten_hamming} is
\begin{equation} \label{eq:elem_sym_poly_state_fourier_raw}
    \big| \tilde{P}^{(k)} \big>
    =
    \frac{1}{\sqrt{\binom{m}{k}}}
    \sum_{\x , \z\in \F_p^n}
    \sum_{\substack{\y \in \F_p^m \\ | \y | = k}}
    \omega_p^{-(B^T \y)^T\x - \x^T (C \cdot \y) \x }
    \omega_p^{\z^T \x}
    \prod_{\substack{i=1\\y_i\ne 0}}^m
    \tilde{g}_{i}(y_i)
    \left| \z \right>
\end{equation}
We focus on the inner sum over $\x$ in \autoref{eq:elem_sym_poly_state_fourier_raw}
\begin{equation} \label{eq:elem_sym_poly_state_fourier_raw_inner_sum}
    \sum_{\x\in \F_p^n}
    \omega_p^{\left(\z - B^T \y \right)^T\x - \x^T (C \cdot \y) \x }
\end{equation}
Let $\mathbf{a}_{\y, \z} = \z - B^T \y$, then the sum becomes
\begin{equation}
    \to
    \sum_{\x\in \F_p^n}
    \omega_p^{\mathbf{a}_{\y, \z}^T\x - \x^T (C \cdot \y) \x }
\end{equation}


Since we assumed $C$ is diagonal, we can rewrite this sum as a product of different registers (each in $\F_p$)
\[
    \sum_{\x \in \F_p^n}
    \omega_p^{\left(\z - B^T\y\right)^T\x - \x^T (C \cdot \y)\x }
    =
    \prod_{j=1}^n
    \sum_{x_j \in \F_p}
    \omega_p^{\left(z_j - (B^T\y)_j\right)x_j - (C \cdot \y)_{jj} x_j^2 }
\]
Let's denote the matrix $D$ such that $D_{ij} := (C_i)_{jj}$ (the $j^{\text{th}}$ element on the diagonal of the $i^{\text{th}}$ matrix). Then $(C \cdot \y)_{jj} = (D^T\y)_j$, so the expression above becomes
\[
    \prod_{j=1}^n
    \sum_{x_j \in \F_p}
    \omega_p^{\left(z_j - (B^T\y)_j\right)x_j - (D^T \y)_{j} x_j^2 }
\]
The sum is precisely of the form of \autoref{claim:gen_quad_sum}, thus we may write it as
\[
    \prod_{j=1}^n
    \left(
    p \delta_{z_j, (B^T \y)_j} \delta_{(D^T \y)_j}
    +
    \omega_p^{-\frac{1}{4}(D^T \y)_j^{-1} \left(z_j - (B^T \y)_j\right)^2} \chi((D^T \y)_j)g(1; p)
    \right)
\]
We can substitute that back into \autoref{eq:elem_sym_poly_state_fourier_raw} and get
\[
    \sum_{\z\in \F_p^n}
    \sum_{\substack{\y \in \F_p^m \\ | \y | = k}}
    \prod_{\substack{i=1\\y_i\ne 0}}^m
    \tilde{g}_{i}(y_i)
    \prod_{j=1}^n
    \left(
    p \delta_{z_j, (B^T \y)_j} \delta_{(D^T \y)_j}
    +
    \omega_p^{-\frac{1}{4}(D^T \y)_j^{-1} \left(z_j - (B^T \y)_j\right)^2} \chi((D^T \y)_j)g(1; p)
    \right)
    \left| z_j \right>
\]

We can now convert the sum over $\z \in \F_p^n$ to a product of $n$ different sums over $z_j \in \F_p$ like so
\[
    \sum_{\substack{\y \in \F_p^m \\ | \y | = k}}
    \prod_{\substack{i=1\\y_i\ne 0}}^m
    \tilde{g}_{i}(y_i)
    \prod_{j=1}^n
    \left(
    p \delta_{(D^T \y)_j}
    \left| (B^T \y)_j \right>
    +
    \chi((D^T \y)_j)g(1; p)
    \sum_{z_j \in \F_p}
    \omega_p^{-\frac{1}{4}(D^T \y)_j^{-1} \left(z_j - (B^T \y)_j\right)^2}
    \left| z_j \right>
    \right)
\]
And conveniently enough, using \autoref{def:F_alpha_very_convinient}, this becomes exactly \autoref{eq:qft_of_elem_sym_poly_state_simplified}.

\end{proof}

\begin{claim}[QFT of target DQI state for max-QUADSAT] \label{claim:fourier_of_DQI_state_max_QUADSAT}
    The QFT of \autoref{def:dqi_state_max_quadsat} is
    \begin{equation} \label{eq:QFT_of_target_DQI_state_max_QUADSAT} 
        \big| \tilde{P}(f) \big>
        =
        \sum_{k=0}^\ell
        \frac{w_k}{\sqrt{\binom{m}{k}}}
        \sum_{\substack{\y \in \F_p^m \\ | \y | = k}}
        \prod_{\substack{i=1\\y_i\ne 0}}^m
        \tilde{g}_{i}(y_i)
        \prod_{j=1}^n
        F_{(B^T\y)_j} \left| (D^T \y)_j \right>
    \end{equation}
\end{claim}
\begin{proof}
    This follows immediately from \autoref{eq:dqi_state_as_sum_of_sym_elem_poly}, \autoref{eq:qft_of_elem_sym_poly_state_simplified}, and the linearity of the QFT.
\end{proof}

\subsection{DQI algorithm for max-QUADSAT} \label{sec:dqi_quantum_algorithm_for_max_quadsat}


We show an efficient algorithm to prepare the DQI state for some instances of max-QUADSAT, as described in \autoref{claim:efficient_algo_for_quadsat}. If all vectors $b_i$ are zero vectors, then the matrix $B$ is the zero matrix, and the QFT of the DQI state for max-QUADSAT \autoref{eq:QFT_of_target_DQI_state_max_QUADSAT} is
\begin{equation}
        \big| \tilde{P}(f) \big>
        =
        \sum_{k=0}^\ell
        \frac{w_k}{\sqrt{\binom{m}{k}}}
        \sum_{\substack{\y \in \F_p^m \\ | \y | = k}}
        \prod_{\substack{i=1\\y_i\ne 0}}^m
        \tilde{g}_{i}(y_i)
        \prod_{j=1}^n
        F_0 \left| (D^T \y)_j \right>
\end{equation}
To prepare this state, we will need to implement $F_0$.
\begin{claim}
    The operator $F_0$ can be efficiently implemented. 
\end{claim}
\begin{proof}
    This follows immediately from \autoref{prim:quadratic_phase} and \autoref{prim:quantum_conditions}.
\end{proof}
Now we have everything we need to prepare the state \autoref{def:dqi_state_max_quadsat} for diagonal quadratic constraints with no linear constraints. The following is a full description of such an algorithm.

\begin{mdframed}[linewidth=0.8pt,roundcorner=5pt] 
\begin{algorithm}[max-QUADSAT state preparation]\label{algo:max_quadsat_diag_quad_no_linear}
\textbf{Step 1 -- Prepare weights.} Where $w_k$ are precomputed classically to give optimal results.
\[
    \sum_{k = 0}^\ell w_k \left| k \right>
\]

\textbf{Step 2 -- Prepare Dicke state.} Use the  $\left| k \right>$ register to prepare the Dicke state of weight $k$ \cite{dicke_states_bartschi2022short}.
\[
    \sum_{k = 0}^\ell w_k \left| k \right>
    \frac{1}{\sqrt{\binom{m}{k}}} \sum_{\substack{\mathbf{\mu} \in \set{0,1}^m \\ | \mathbf{\mu}| = k}}
        \left| \mathbf{\mu} \right> 
\]

\textbf{Step 3 -- Uncompute $\left| k \right>$}, using the $\left| \mu \right>$ register, and $k = | \mu |$
\begin{equation} \label{eq:max_quadsat_algo_nolinear_commuting_dicke_state}
    \sum_{k = 0}^\ell w_k
    \frac{1}{\sqrt{\binom{m}{k}}} \sum_{\substack{\mathbf{\mu} \in \set{0,1}^m \\ | \mathbf{\mu}| = k}}
        \left| \mathbf{\mu} \right> 
\end{equation}

\textbf{Step 4 -- Create error register in prime field.} Let $G_i$ denote the transform on $\set{\left| 0 \right>, \left| 1 \right>}$ such that $G_i \left| 0 \right> = \left| 0 \right>$ and $G_i \left|1\right> = \sum_{y \in \F_p} \tilde{g}_i(y) \left| y \right>$ 
then apply $G = \prod_{i=1}^m G_i$ to \autoref{eq:max_quadsat_algo_nolinear_commuting_dicke_state}
\begin{equation}
    \sum_{k = 0}^\ell w_k
    \frac{1}{\sqrt{\binom{m}{k}}} \sum_{\substack{\y \in \F_p^m \\ | \y | = k}}
        \prod_{\substack{i = 1\\ y_i \ne 0}}^m \tilde{g}_i(y_i)\left| \y \right>
\end{equation}

\textbf{Step 5 -- Compute syndrome.} Reversibly compute $D^T\y = C_1 y_1 + \cdots + C_m y_m$ into another register 
\begin{equation}
    \sum_{k = 0}^\ell w_k
    \frac{1}{\sqrt{\binom{m}{k}}} \sum_{\substack{\y \in \F_p^m \\ | \y | = k}}
        \prod_{\substack{i = 1\\ y_i \ne 0}}^m \tilde{g}_i(y_i)
        \left| \y \right>
        \big| D^T\y \big>
\end{equation}


\textbf{Step 6 -- Uncompute error register by decoding.} Calculating $\left| \y \right>$ from $\left| D^T \y \right>$ is equivalent to a decoding problem with $| \y | = k$ errors. Assuming we can efficiently solve this, we can uncompute the $\left| \y \right>$ register.
\begin{equation}
    \sum_{k = 0}^\ell w_k
    \frac{1}{\sqrt{\binom{m}{k}}} \sum_{\substack{\y \in \F_p^m \\ | \y | = k}}
        \prod_{\substack{i = 1\\ y_i \ne 0}}^m \tilde{g}_i(y_i)
        \big| D^T\y \big>
\end{equation}


\textbf{Step 7 -- Apply $F_0$ on each subregister of the syndrome register.}
\[
        \sum_{k=0}^\ell
        \frac{w_k}{\sqrt{\binom{m}{k}}}
        \sum_{\substack{\y \in \F_p^m \\ | \y | = k}}
        \prod_{\substack{i=1\\y_i\ne 0}}^m
        \tilde{g}_{i}(y_i)
        \prod_{j=1}^n
        F_0 \left| (D^T \y)_j \right>
\]
{\color{red} Mistake disclaimer: Applying $F_0$ to each sub-register involves a measurement that succeeds with probability $\frac{1}{8}$, hence the probability that this step succeeds is $\frac{1}{8^m}$.}

\textbf{Step 8 -- Apply QFT.} According to \autoref{claim:fourier_of_DQI_state_max_QUADSAT}, this yields the desired state \autoref{def:dqi_state_max_quadsat}
\begin{equation}
    \sum_{\x \in \F_p^n} P(f(\x)) \left| \x \right>
\end{equation}

\end{algorithm}
\end{mdframed}

The existence of this algorithm proves \autoref{claim:efficient_algo_for_quadsat}.
\section{Instances of max-QUADSAT} \label{sec:problems}

In this section, we will discuss problems that can be optimized using max-QUADSAT. Note that any max-LINSAT problem has a corresponding max-QUADSAT variant, if we take the $\bb_i$ vectors of max-LINSAT and turn them into the diagonals of the $C_i$ matrices from max-QUADSAT. To demonstrate this, we will use the following example.

\subsection{Quadratic Optimal Polynomial Intersection}


The OPI problem was considered in \cite{og_dqi}, and was the strongest demonstration of quantum speedup they achieved. Inspired by that, we defined the quadratic-OPI problem \autoref{def:quadratic_opi}, which adds a condition that the coefficients of the polynomial are quadratic residues.



It is worth mentioning that quadratic OPI is at least as hard as regular OPI, that is, because solutions of quadratic OPI are also solutions of regular OPI, but the inverse is not true in general. We do not know of a trivial way to use standard DQI to gain algorithmic speedup on this problem.

Now we can show that quadratic OPI is a special case of max-QUADSAT, similarly to how that was shown for OPI and max-LINSAT; this would prove \autoref{theorem:quadratic_opi_is_max_quadsat}.

\begin{proof}[Proof Theorem 2]
    Let $q_0, q_1, \dots, q_{n-1} \in \F_p$ be the coefficients of the polynomial $Q$:
    \[
        Q(y) = \sum_{j=0}^{n-1} q_j y^j
    \]
    Let $\gamma$ be a primitive root of $\F_p$, then every element of $\F_p$ can be written as a power of $\gamma$, and we can write the OPI objective function as:
    \[
        f_{\text{OPI}}(Q) = \left|\set{i \in \set{0, . . . , p - 2} : Q(\gamma^i) \in \F_{\gamma^i} }\right|
    \]
    Now, define each DQI constraint function $f_i$ as
    \[
        f_i(x) = \begin{cases}
            +1 & \text{if $x \in \F_{\gamma^i}$} \\
            -1 & \text{otherwise}
        \end{cases}
    \]
    And define the max-QUADSAT matrices $C_i$ to be the diagonal matrices
    \[
        (C_i)_{jj} = \gamma^{i \times j}\quad i=0,\dots,p-2 \quad j=0,\dots,n-1
    \]
    Then the max-QUADSAT objective function is $f(\mathbf{q}) = \sum_{i=0}^{p-2} f_i (\mathbf{q}^T C_i \mathbf{q})$, where $\mathbf{q}$ is a vector of coefficients. Now notice that $\mathbf{q}^T C_i \mathbf{q} = \sum_{j=0}^{n-1} (C_i)_{jj} q_j^2 = \sum_{j=0}^{n-1} \gamma^{i \times j} q_j^2 = Q^*(\gamma^i)$, where $Q^*$ is a polynomial with the coefficients $q_0^2, \dots, q_{n-1}^2$. Hence, $f(\mathbf{q}) = 2 f_{\text{OPI}}(Q^*) - (p-1)$, which means this instance of max-QUADSAT maximizes the same objective function of the OPI problem, where the coefficients used are the squares of $\mathbf{q}$. Lastly, note that all conditions of \autoref{claim:efficient_algo_for_quadsat} hold, with the code being a Reed-Solomon code of distance $n+1$, which proves that this state can be efficiently prepared by \autoref{algo:max_quadsat_diag_quad_no_linear}.
\end{proof}
Note that quadratic OPI with our variant of DQI is equivalent to regular OPI with regular DQI, with the same quantum advantage guarantees. So, for example, if $n \simeq p/10$, $r = \lfloor p/2 \rceil$, then the fraction of satisfied constraints by measuring the DQI state is $1/2 + \sqrt{19}/20 \approx 0.7179$, which is more than $0.55$ that is achieved classically with currently known techniques as discussed in \cite{og_dqi}. To get higher fractions classically, one must fall back to brute force, which runs in exponential time.

\section{Optimal expected fraction of satisfied constraints} \label{sec:optimal_expected_fraction_of_satisfied_constraints}

Now that we know how to prepare the DQI state for max-QUADSAT, we can analyze the expected number of satisfied constraints in this state. We now re-derive the semicircle law of \cite{og_dqi} in a different way that would allow us to extend it to max-QUADSAT.

\begin{claim}[DQI satisfied constraints distribution] \label{claim:dqi_state_is_from_N_s}
    For the DQI state \autoref{eq:dqi_state}, the probability to measure a string $\x$ that satisfies $s$ constraints is
    \[
        \text{Pr}_{\text{\tiny DQI}}[s] = \left| P(2s - m) \right|^2 N_s(s)
    \]
    Where $P$ is the DQI polynomial, and $N_s(s)$ is the probability for a random $\x \in \F_p^n$ to satisfy $s$ constraints.
\end{claim}
\begin{proof}
We consider the DQI state \autoref{eq:dqi_state}
\[
    \left| P(f) \right>
    =
    \sum_{\x \in \F_p^n} P(f(\x)) \left| \x \right>
\] 
First, we note that $f(x)$ takes values only from $-m$ to $m$, so we will separate the sum as follows:
\[
    \sum_{\x \in \F_p^n} P(f(\x)) \left| \x \right>
    =
    \sum_{f = -m}^m \sum_{\substack{\x \in \F_p^n \\ \text{where} \\ f(\x) = f}}
    P(f) \left| \x \right>
    =
    \sum_{f = -m}^m
    P(f)
    \sum_{\substack{\x \in \F_p^n \\ \text{where} \\ f(\x) = f}}
    \left| \x \right>  
\]
Thus we can denote by $N_f(f)$ the number of $\x \in \F_p^n$ such that $f(\x) = f$ (divided by $p^n$ so it'd be a probability density), then the probability of measuring $\x$ with $f(\x) =f$ in the DQI state is
\[
    \text{Pr}_{\text{\tiny DQI}}[f] = \left| P(f) \right|^2 N_f(f)
\]
It will be more convenient to work with the quantity "number of satisfied constraints", denoted $s$ ($f = 2s - m$). 
Then if we denote by $N_s(s)$ the probability for $\x \in \F_p^n$ to satisfy $s$ constraints, the probability to measure $\x$ with $s$ satisfied constraints from the DQI state is
\[
    \text{Pr}_{\text{\tiny DQI}}[s] = \left| P(2s - m) \right|^2 N_s(s)
\]
\end{proof}

\begin{remark} \label{remark:all_dep_is_in_N_s}
     All the dependence on the actual problem (LINSAT / QUADSAT / constraints / etc...) all enter via $N_s(s)$. Therefore, showing $N_s(s)$ of max-QUADSAT is the same as $N_s(s)$ of max-LINSAT is enough to use Theorem 4.1 of \cite{og_dqi} (semicircle law for DQI) to claim the expected number of satisfied constraints of max-QUADSAT is the same as that of max-LINSAT.
\end{remark} 

We shall now show that $N_s(s)$ is distributed \textit{"close enough"} to a binomial distribution with mean $m \frac{r}{p}$ and variance $m \frac{r}{p}\left(1 - \frac{r}{p}\right)$, so that the mean of $s$ in  $\text{Pr}_{\text{\tiny DQI}}[s]$ is the same as if we'd use a perfect binomial distribution for $N_s(s)$.

\begin{claim}[matrix functional leaves uniform distribution uniform] \label{claim:linsat_is_uniform}
    Let $B \in \F_p^{m \times n}$ be a matrix where $m < n$ and $\text{rank}(B) = m$, and let $\x \in \F_p^n$ be a random variable uniformly distributed over $\F_p^n$. Then $\vv = B \x$ is uniformly distributed over $\F_p^m$.
\end{claim}
\begin{proof}
    Consider $B$ as a linear map from $\F_p^n$ to $\F_p^m$, then by the isomorphism theorem $\F_p^n / \ker{B} \simeq \text{im} B$. Because $B$ has rank $m$ it spans $\F_p^m$, so it's surjective $\text{im } B = \F_p^m$. Thus, the size of the preimage of each $\vv \in \F_p^m$ is the same (it's equal to $\left| \ker{B} \right| = p^{n-m}$). So if $\x$ is uniformly distributed, then so is $\vv$.
\end{proof} 
\begin{claim}[LINSAT is close to binomial distribution] \label{claim:dqi_dist_like_binom}
Let $B \in \F_p^{m \times n}$ be a matrix such that $B^T$ is a parity-check matrix of a code with distance $d > 2\ell + 1$,\footnote{Having $d > 2\ell +1$ rather than $d \ge 2\ell + 1$ makes the argument more exact and simpler, but it is not strictly necessary.} and let $f_i : F_p \to \pm 1$ be constraint functions for an instance of max-LINSAT. Suppose that $\left| f_i^{-1} (+1) \right| = r$ for all $i=1,\dots,m$ and some $r \in \set{1, \dots, p-1}$.

Let $\x \in \F_p^n$ be a random variable uniformly distributed in $\F_p^n$, and let $s_i$ be a random variable such that $s_i = 0$ if $f_i (\bb_i \cdot \x) = -1$, and $s_i = +1$ otherwise\footnote{We can write this as $s_i = \frac{1+f(\bb_i \cdot \x)}{2}$}. Let $S = s_1 + \cdots + s_m$, and let $N(s)$ denote the probability that $S=s$. Let $N_{\text{binom}}(s)$ be a binomial distribution with mean $m \frac{r}{p}$ and variance $m \frac{r}{p} (1 - \frac{r}{p})$.

Then for every (properly normalized\footnote{Assuming the coefficients of $P(s)$ are chosen such that $P(s)$ is properly normalized $\sum_s P(s)N(s) = 1$, then also $\sum_s P(s)N_{\text{binom}}(s)=1$. This follows from a weaker version of the same proof of \autoref{claim:dqi_dist_like_binom}}) polynomial $P(s)$ of degree $\le 2\ell$: $\left< s \right>_{s \sim P(s) N(s)} = \left< s \right>_{s \sim P(s) N_{\text{binom}}(s)}$
\end{claim}
\begin{proof} 
    $B^T$ is a parity-check matrix of a code with distance $d > 2\ell + 1$, thus every subset of $2\ell+1$ columns in $B^T$ is linearly independent\footnote{If $B^T$ can correct up to $\lfloor {\frac{d - 1}{2}} \rfloor > \ell$. By definition of code distance, every non-zero $\y \in \F_p^m$ with hamming weight $|\y| < d$ satisfies $B^T\y \ne \mathbf{0}$. Or in other words, every subset of at most $d-1 \ge 2\ell + 1$ (notice $\ge$ instead of $>$) columns of $B^T$ is linearly independent}. and hence each subset of $2 \ell+1$ rows in $B$ is linearly independent. Thus, by \autoref{claim:linsat_is_uniform} every subset of $2\ell+1$ rows in $B^T\x$ is uniformly distributed in $\F_p^{2\ell+1}$, and hence independent. Applying $f_i$ to the $i^{\text{th}}$ rows of the vector $B^T \x$ immediately results in the fact that any collection of $2\ell+1$ variables in $\set{s_i}$ are independently distributed Bernoulli variables with mean $\frac{r}{p}$ and variance $\frac{r}{p} \left(1 - \frac{p}{r}\right)$.
    
    We shall look at the $k$ ($\le 2\ell + 1$) moment of $S = s_1 + \cdots + s_m$, given by
    \begin{equation} \label{eq:S_k_moemnt}
        \left< S^k \right> = \left< \left( s_1 + \cdots + s_m \right)^k \right>
    \end{equation}
    And compare it to $S_{\text{binom}} = s_1^{\text{binom}} + \cdots + s_m^{\text{binom}}$ ($\set{s_1^{\text{binom}}, \dots, s_m^{\text{binom}}}$ are all i.i.d Bernoulli variables with success probability $\frac{r}{p}$). The $k^{\text{th}}$ moment of $S_{\text{binom}}$ is
    \begin{equation} \label{eq:S_binom_k_moment}
        \left< S_{\text{binom}}^k \right> = \left< \left( s_1^{\text{binom}} + \cdots + s_m^{\text{binom}} \right)^k \right>
    \end{equation}
    We note that every element in \autoref{eq:S_k_moemnt} and \autoref{eq:S_binom_k_moment} contains up to $k$ variables from $s_i$, and since we assume $k \le 2 \ell+1$, these $k$ variables are independent. Thus the first $2\ell+1$ moments of $S$ and $S_{\text{binom}}$ match
    \begin{equation}
       \forall k \le 2\ell+1 \ : \ \left< S^k \right> = \left< S_{\text{binom}}^k \right>
    \end{equation}
    Hence, their mean is given by
    \[
        \left< s \right>_{s \sim P(s)N(s)}
        =
        \sum_{s} s P(s) N(s)
        =
        \sum_s \sum_{k=1}^{2\ell + 1} \alpha_k s^k N(s)
        =
        \sum_{k=1}^{2\ell + 1} \alpha_k \underbrace{\sum_s s^k N(s)}_{\left<S^k\right>}
        =
        \sum_{k=1}^{2\ell + 1} \alpha_k \left<S^k\right>
    \]
    And similarly
    \[
        \left< s \right>_{s \sim P(s)N_{\text{binom}}(s)}
        =
        \sum_{k=1}^{2\ell + 1} \alpha_k \underbrace{\sum_s s^k N_{\text{binom}}(s)}_{\left<S^k\right>}
        =
        \sum_{k=1}^{2\ell + 1} \alpha_k \left<S^k\right>
        =
        \left< s \right>_{s \sim P(s)N(s)}
    \]
    Which is what we wanted to show.
\end{proof}

Now, thanks to \autoref{remark:all_dep_is_in_N_s}, we can go on to show the DQI state for max-QUADSAT also obeys the semicircle law. Before that, let's re-derive Lemma 9.2 (and thus Theorem 4.1) of \cite{og_dqi} in a completely different way. Now we have all we need to prove \autoref{claim:rederive_lemma92}


\begin{proof}[Proof Thoerem 3]
    Using \autoref{claim:dqi_dist_like_binom}, the expectation value of $s$ in the DQI state is
    \[
    \left< s \right>
    =
    \frac{
    \sum_{s = 0}^m s \left| P(s) \right|^2 N(s)
    }{
    \sum_{s = 0}^m \left| P(s) \right|^2 N(s)
    }
    \]
    Where $N(s)$ is a binomial distribution of $m$ Bernoulli variables, each with a success probability $r/p$.
    We define the following inner product with respect to a binomial distribution weight
    \[
        \left< P, Q \right>_{\text{binom}}
        =
        \sum_{s = 0}^m P^*(s) Q(s) N(s)
    \]
    Then the expectation value can be written as the product of $P$ and $sP$
    \[
        \left< s \right>
        =
        \frac{
        \left< s P, P \right>_{\text{binom}}
        }{
        \left< P, P \right>_{\text{binom}}
        }
    \]
    It is a known fact that the \textit{Krawtchouk polynomials} are orthogonal with respect to a binomial weight \cite{krawtchouk_koekoek1996askeyschemehypergeometricorthogonalpolynomials}. That is, let $K_i(s)$ be the $i^{\text{th}}$ (normalized) Krawtchouk polynomial\footnote{In the notation of \cite{krawtchouk_koekoek1996askeyschemehypergeometricorthogonalpolynomials}, $N \to m,\ p \to r/p,\ x \to s,\ K_i \to K_i / \sqrt{\left< K_i, K_i \right>_{\text{binom}}}$}, then
    \[
        \left< K_i, K_j \right>_{\text{binom}} = \delta_{i, j}
    \]
    Now let's write the DQI polynomial $P(s)$ as a linear combination of Krawtchouk polynomials
    \[
        P(s)
        =
        \sum_{k=0}^\ell
        w_k K_k(s)
    \]
    First, we see that $\left< P, P \right>_{\text{binom}} = \ww^\dagger \ww$, where we denoted $\ww = \left(w_0, \cdots,  w_\ell\right)^T$, so we will limit ourselves to $\ww^\dagger \ww = 1$ from now. Then we are left with
    \[
        \left< s \right>
        =
        \big< s \sum_{k=0}^\ell w_k K_k, \sum_{k'=0}^\ell w_{k'} K_{k'} \big>_{\text{binom}}
        =
        \sum_{k, k'=0}^\ell w_k^* w_{k'} \big< s K_k, K_{k'} \big>_{\text{binom}}
    \]
    Second, we can invoke the three-term recurrence relation for the Krawtchouk polynomials \cite{krawtchouk_koekoek1996askeyschemehypergeometricorthogonalpolynomials}
    \[
        sK_k =
          \underbrace{[\frac{r}{p}(m-k) + k(1-\frac{r}{p})]}_{a_{k0}}K_k
        - \underbrace{\sqrt{\frac{r}{p}(1-\frac{r}{p})(m-k)(k+1)}}_{a_{k+}}K_{k+1} 
        - \underbrace{\sqrt{\frac{r}{p}(1-\frac{r}{p})(m-k+1)k}}_{a_{k-}}K_{k-1} 
    \]
    Then
    \begin{equation}\label{eq:s_expectation_before_linalg}
        \left< s \right>
        =
        \sum_{k, k'=0}^\ell w_k^* w_{k'} \big< a_{k0}K_k - a_{k+}K_{k+1} - a_{k-}K_{k-1}, K_{k'} \big>
        =
        \sum_{k=0}^\ell a_{k0}\left| w_k \right|^2 - a_{k+} w_k^*w_{k+1} - a_{k_-} w_k^*w_{k-1}
    \end{equation}
    Now we can write $a_{k0}, a_{k\pm}$ in a more familiar way to make the parallel to Lemma 9.2 of \cite{og_dqi} clear:
    \[
        a_{k0}
        = \frac{mr}{p} + k\frac{\sqrt{r(p-r)}}{p}\underbrace{\frac{p- 2r}{\sqrt{r(p-r)}}}_{d}
        \qquad ; \qquad
        a_{k-} = a_{(k-1)+}
        =
        \frac{\sqrt{r(p-r)}}{p}
        \underbrace{\sqrt{k(m-k+1)}}_{a_k}
    \]
    Then if we denote $d := \frac{p- 2r}{\sqrt{r(p-r)}}$, $a_k := \sqrt{k(m-k+1)}$, and the matrix
    \[
        A
        =
        \begin{pmatrix}
            0 & -a_{1} \\
            -a_{1} & d & -a_{2} \\
            & -a_{2} & 2d & \ddots \\
            & & \ddots  & & -a_{\ell}\\
            &&& -a_{\ell} & \ell d
        \end{pmatrix}
    \]
    Then we can rewrite \autoref{eq:s_expectation_before_linalg} as
    \[
        \left< s \right> = \frac{mr}{p} + \frac{\sqrt{r(p-r)}}{p} \ww^\dagger A \ww
    \]
    This becomes exactly \autoref{eq:lemma92_result} (without negative signs on the off-diagonal) when we transform $\ww$ to a basis where $w_k$ is multiplied by $-1$ for even $k$. 
\end{proof}

\subsection{Application to max-QUADSAT} \label{sec:opti_app_to_max_quadsat}

So far, in \autoref{claim:dqi_state_is_from_N_s} we have shown that the expectation of the DQI state $\left< s \right>$ depends only on the distribution $N(s)$ of the number of constraints satisfied by a random assignment. In \autoref{claim:dqi_dist_like_binom} we have shown that $N(s)$ is distributed "close enough" to a binomial distribution such that they are indistinguishable when we take the mean.

A key part of that proof was \autoref{claim:linsat_is_uniform}, where we showed that $B\x$ is uniformly distributed for uniformly distributed $\x$. Unfortunately, this claim breaks if we replace the linear functional $B\x$ with a quadratic functional $\x^T C \x$. That is, for a uniform $\x \in \F_p^n$, $\x^T C \x$ is not uniform in $\F_p$. But hope is not lost, because instead, we can show that $\x^T C \x$ is exponentially close to being uniformly distributed. We will now prove the contents of \autoref{claim:quadsat_is_almost_uniform_intorduction}.

\begin{proof}[Proof Claim 2]
Denote by $\lambda_i$ the non-zero entries on the diagonal of $C$, then
\[
    \x^T C \x = \sum_{i=1}^r \lambda_i x_i^2
\]
let $a \in \F_p$, the probability of $\x^T C \x$ to equal $a$ is
\[
    \text{Pr}(\x^T C \x = a)
    =
    \frac{1}{p^n}\sum_{\x \in \F_p^n} \delta_{\x^T C \x, a}
    =
    \frac{1}{p^n}\sum_{\x \in \F_p^n} \frac{1}{p}\sum_{u \in \F_p}\omega_p^{u(\x^T C \x - a)}
    =
    \frac{1}{p^{n+1}}
    \sum_{u \in \F_p}
    \omega_p^{-ua}
    \sum_{\x \in \F_p^n}
    \omega_p^{u\sum_{i=1}^r \lambda_i x_i^2}
\]
Now we can recognize that the inner sum is just a product of $r$ quadratic gauss sums \autoref{def:quadratic_gauss_sum}, $g(u \lambda_i; p)$
\[
    \sum_{\x \in \F_p^n}
    \omega_p^{u \sum_{i=1}^n \lambda_i x_i^2}
    =
    \prod_{i=1}^r
    \sum_{x_i \in \F_p}
    \omega_p^{u\lambda_i x_i^2}
    =
    \prod_{i=1}^r
    g(u\lambda_i; p)
\]
Now using \autoref{eq:quadratic_gauss_sum_a} we can write $g(u\lambda_i; p) =p \delta_{u\lambda_i, 0} + i_p \chi(u \lambda_i) \sqrt{p}$ and since $\lambda_i \ne 0$ the entire product is $p^n \delta_{u, 0} + i_p^r \chi(u)^r \chi\left(\prod_{i=1}^r \lambda_i\right) p^{r/2}$. Plugging that back in yields
\[
    \text{Pr}(\x^T C \x = a)
    =
    \frac{1}{p^{r+1}}
    \sum_{u \in \F_p}
    p^r \delta_{u, 0}
    +
    \chi\left(\prod_{i=1}^r \lambda_i\right)
    i_p^r
    \omega_p^{-ua}
    \chi(u)^r
    p^{r/2}
    =
    \frac{1}{p}
    +
    \frac{i_p^r \chi\left(\prod_{i=1}^r \lambda_i\right)}{p^{r/2+1}}
    \sum_{u \in \F_p}
    \omega_p^{-ua}\chi(u)^r
\]
For even $r$, $\chi(u)^r = 1-\delta_{u, 0}$, and the sum becomes $\sum_{u \in \F_p^*}\omega_p^{-ua} = \sum_{u \in \F_p}\omega_p^{-ua} - \omega_p^{-0\cdot a} = p\delta_{a,0}-1$. For odd $r$ then $\chi(u)^r = \chi(u)$, and the sum is just a quadratic gauss sum $\sum_{u \in \F_p} \omega_p^{-ua}\chi(u) = g(-a; p) = i_p \chi(-a) \sqrt{p}$. Hence
\[
    \text{Pr}(\x^T C \x = a)
    =
    \frac{1}{p}
    +
    \frac{\chi\left(\prod_{i=1}^r \lambda_i\right)}{p^{r/2+1}}
    \begin{cases}
        i_p^r\left(p \delta_{a,0} - 1\right) & r \text{ even} \\
        i_p^{r+1} \chi(-a) \sqrt{p} & r \text{ odd}
    \end{cases}
\]
Notice that $i_p$ is always raised to an even power, so it's never imaginary (can only be $\pm 1$). And recall that $\chi$ only takes the values $\pm 1, 0$. So the largest the second term can be (in absolute value) is $\frac{p-1}{p^{r/2+1}} = p^{-r/2} - p^{-r/2-1}$, which occurs when $a=0$ and $r$ is even. This proves \autoref{eq:quadsat_is_almost_uniform_introduction}.
\end{proof}
In the limit $r \to \infty$, $\x^T C \x$ is distributed uniformly, just like the linear case $B\x$ in \autoref{claim:linsat_is_uniform}, thus all conditions of \autoref{claim:dqi_dist_like_binom} hold for max-QUADSAT if we replace the linear functional $\bb_i \cdot \x$ with a quadratic one $\x^T C_i \x$, and hence \autoref{claim:quadsat_also_binomial} is proven and the semicircle law holds for the DQI state of max-QUADSAT.

\appendix

\section{Quantum Techniques} \label{sec:appendix_quantum_techniques}

The goal of this section is to create "primitives" (quantum subroutines) that would be useful for us in preparing the max-QUADSAT states, see \autoref{sec:dqi_quantum_state_for_max_quadsat}.

\begin{primitive}[Quadratic phase] \label{prim:quadratic_phase}
    Let $p$ be an odd prime, and let $a \in \F_p$. There exists an efficient quantum algorithm that succeeds with probability $\frac{1}{8}$
    to do the following transformation
    \begin{equation} \label{eq:quadratic_phase_transform}
        \left| a \right>
        \to
        \sum_{x \in \F_p} \omega_p^{a x^2} \left| x \right>
    \end{equation}
    We will assume here that $p = 3 \mod 4$ to not separate everything into cases where $-1$ is a residue or nonresidue, but the proof for $p= 1 \mod 4$ is essentially identical.
\end{primitive}

\begin{proof}
We begin with a definition 

\begin{definition}[invertible square root in a finite field] \label{def:invertible_sqrt_finite_field}
    Let $x \in \F_p$ be a residue, denote $\sqrt[\pm]{x}$ as the two square roots of $x$, thus $(\sqrt[+]{x})^2 = (\sqrt[-]{x})^2 = x$ and $\sqrt[+]{x} \ne \sqrt[-]{x}$. 
    Then we define the operation $\sqrt{x}$ as
    \[
        \sqrt{x}
        :=
        \begin{cases}
            0 & \text{if $x=0$} \\
            \sqrt[+]{x} & \text{if $x$ is a residue} \\
            \sqrt[-]{-x} & \text{if $x$ is a nonresidue}
        \end{cases}
    \]
\end{definition}

Intuitively, exactly half of the elements in $\F_p^\times$ are residues, and each residue has two roots, so we "take" one root from each residue and "give" it to a non-residue.
It's easy to verify that $\sqrt{x}$ is a bijection on $\F_p$ onto itself.

We can easily perform the following transformation
\[
    \left| a \right>
    \xrightarrow{\text{QFT}}
    \sum_{x \in \F_p} \omega_p^{a x} \left| x \right>
    \xrightarrow{\sqrt{x}}
    \sum_{x \in \F_p} \omega_p^{a x} \left| \sqrt{x} \right>
\]
Where doing $\left| x \right> \to \left| \sqrt{x} \right>$ is done by reversibly computing $\sqrt{x}$ in an ancilla\footnote{This requires an efficient algorithm to take the square root of a quadratic residue in a finite field, see \cite{sqrt_finite_adiguzelgoktas2024squarerootcomputationfinite} for example}, then uncomputing $x$ by squaring.

Let us denote $s = \sqrt{x}$, then
\[
    x = \begin{cases}
        s^2 & \text{if $s$ is a residue} \\
        -s^2 & \text{if $s$ is a nonresidue}
    \end{cases}
\]
Or, more compactly, $x = \chi(s)s^2$. Then, the sum above becomes
\[
    \sum_{s \in \F_p} \omega_p^{\chi(s)a s^2} \left| s \right>
\]
Using this method, we can implement the following two unitary operators\footnote{$U_-$ is implemented by using IQFT instead of QFT in the first step}
\[
    U_{\pm} \left| a \right>
    =
    \sum_{s \in \F_p} \omega_p^{\pm\chi(s) a s^2} \left| s \right>
\]
For $U_+$, if we reversibly compute $\left| \chi(s) \right>$ into an ancilla register, measure, and postselect on $+$, we will get the state
\[
    \sum_{\substack{s \in \F_p \\ \text{$s$ is a residue}}} \omega_p^{a s^2} \left| s \right>
\]
And similarly for $U_-$, we reversibly compute $\left| \chi(s) \right>$ into an ancilla register, measure, and postselect on $-$, we will get the state
\[
    \sum_{\substack{s \in \F_p \\ \text{$s$ is a nonresidue}}} \omega_p^{a s^2} \left| s \right>
\]
To combine the two, we will use an approach very similar to the Linear Combination of Unitaries method (LCU, see \cite{childs2017lecturenotes} section 27.3).


We first define the "select" operator
\[
    S \left| 0/1 \right> \left| \psi \right>
    =
    \left| 0/1 \right> U_{\pm}\left| \psi \right>
\]
Then
\[
    S\left(\left| 0 \right> + \left| 1 \right>\right) \left| a \right>
    =
    \left| 0 \right>
        \sum_{s \in \F_p}
        \omega_p^{\chi(s) a s^2}   
        \left| s \right>
    +
    \left| 1 \right>
        \sum_{s \in \F_p}
        \omega_p^{-\chi(s) a s^2}
        \left| s \right>
\]
We can now compute $\left| \chi(s) \right>$ into two ancilla registers, one conditioned on the first register being $\left| 0 \right>$, and the other conditioned on it being $\left| 1 \right>$. Then post select on the first $\chi(s)$ being $+$ and the second being $-$ ($1/4$ chance), and get the state
\[
    \left| 0 \right>
        \sum_{\substack{s \in \F_p\\ \text{residue}}}
        \omega_p^{ a s^2}
        \left| s \right>
    +
    \left| 1 \right>
        \sum_{\substack{s \in \F_p\\ \text{nonresidue}}}
        \omega_p^{ a s^2}
        \left| s \right> 
\]
Now do a Hadamard on the first register, and post-select $0$. This would yield exactly the desired result.

\end{proof}

\begin{primitive}[shifted quadratic phase] \label{prim:shifted_quadratic_phase}
    Let $a, b \in \F_p$, then there exists an efficient quantum algorithm that succeeds in doing the following transformation with probability $\frac{1}{p}$:
\[
    \left| a \right> \left| b \right>
    \to
    \sum_{x \in \F_p} \omega_p^{a (x + b)^2} \left| x \right>
\]
\end{primitive}
\begin{proof}
    Using \autoref{prim:quadratic_phase}, we can efficiently implement the following transformation
    \[
        \left| a \right> \left| b \right>
        \to
        \sum_{x \in \F_p}
            \omega_p^{a x^2} \left| x \right> \left| b \right>
    \]
    Then we can reversibly subtract the second register from the first register and get the following:
    \[
        \sum_{x \in \F_p}
            \omega_p^{a x^2} \left| x - b \right> \left| b \right>
    \]
    Which is equal to (after change of variables, $x \to x + b$)
    \[
        \sum_{x \in \F_p}
            \omega_p^{a (x+b)^2} \left| x\right> \left| b \right>
    \]
    Now all we need is to uncompute the second register, $\left| b \right>$. We can do the usual Hadamard and post-selection trick
    \[
        \sum_{x \in \F_p}
            \omega_p^{a (x+b)^2} \left| x\right>
            \sum_{y \in \F_p} (-1)^{y \cdot b} \left| y \right>
    \]
    Then, measuring the second register and post-selecting on $y = 0$ yields the desired result. Note that post-selection succeeds with probability $\frac{1}{p}$.
\end{proof}

\begin{primitive}[Quadratic form phase for diagonal matrices] \label{prim:quadratic_form_phase_diagonal}
    Let $n$ be an integer, and let $D \in \F_p^{n \times n}$ be a diagonal matrix. There exists an efficient quantum algorithm that does the following transform
    \begin{equation} \label{eq:quadratic_form_phase_diagonal}
        \left| D \right>
        \to
        \sum_{x \in \F_p^n} \omega_p^{\x^T D \x} \left| \x \right>
    \end{equation}
\end{primitive}
\begin{proof}
The target state in \autoref{eq:quadratic_form_phase_diagonal} is just a product of $n$ separate registers, each prepared in exactly the same way as \autoref{prim:quadratic_phase}.
\[
    \prod_{i=1}^n \left(
        \sum_{x \in \F_p} \omega_p^{D_i x_i^2} \left| x_i \right>
    \right)
\]
So we can prepare each sub-register separately and succeed with each time with probability $\ \frac {1}{8}$. If preparing the sub-register failed (i.e., we measured the wrong value for postselection), we can discard that specific register and try again.

\end{proof}

\begin{primitive}[quantum conditions] \label{prim:quantum_conditions}
    Let $P: \F_p \to \set{0,1}$ be an efficiently computable predicate. And assume there are efficient quantum algorithms that does $U_0: \left| x \right> \to \left| f_0(x) \right>$ and $U_1: \left| x \right> \to \left| f_1(x) \right>$. Then there exists an efficient quantum algorithm that does
    \[
        \left| x \right>
        \to
        \delta_{P(x), 0} \left| f_0(x) \right>
        +
        \delta_{P(x), 1} \left| f_1(x) \right>
    \]
\end{primitive}
\begin{proof}
    The final state is equivalent to
    \[
        \left| f_{P(x)}(x) \right>
    \]
    Hence, we can compute the predicate into a single qubit ancilla.
    \[
        \left| x \right>
        \to
        \left| x \right> \left| P(x) \right>
    \]
    We can now apply $U_0$ conditioned on $P(x)=0$, and $U_1$ conditioned on $P(x)=1$ to get
    \[
        \to
        \left| f_{P(x)}(x) \right> \left| P(x) \right>
    \]
    To uncompute the second register, will use the usual Hadamard-postselection trick
    \[
        \xrightarrow{H}
        \left| f_{P(x)}(x) \right> \left(\left| 0 \right> + (-1)^{P(x)} \left| 1 \right>\right)
    \]
    Now we can measure the second register and post-select on $0$ to achieve the desired result
\end{proof}

\section{Gauss Sums} \label{sec:appendix_gauss_sums}

Here, we review some well-known properties of Gauss Sums that are used throughout this paper.

\begin{definition}[Gauss sums over Finite Fields]\label{def:gauss_sums_over_finiite_fields}
     Given a finite field $\F_{p^r}$, a multiplicative character in the field $\chi_\alpha$, and an additive character $e_\beta$, the Gauss sum is
     \[
        G(\F_{p^r}, \chi_\alpha, e_\beta)
        = \sum_{x \in \F_{p^r}} \chi_\alpha(x) e_\beta(x)
        = \sum_{x \in \F_{p^r}} \omega_{p^r -1}^{\alpha \log_g x} \omega_p^{\text{Tr}(\beta x)}
     \]
\end{definition}
Where $\log_g x$ is the discrete logarithm of $x$ in the multiplicative group (since the multiplicative group $\F_{p^r}^\times$ is cyclic, we can write each element $x \in \F_{p^r}^\times$ as $x = g^{n}$, where $g$ is a generator of the group, then $\log_g x = n$). We sometimes use the shorthand notation $G(\F_{p^r}, \alpha, \beta)$ or even $G(\alpha, \beta)$ when it's clear from the context.

Since $\chi_\alpha$ and $e_\beta$ are homomorphisms, we can easily prove the following fact:

\begin{claim}
    For $\beta \ne 0$, it holds that $G(\F_p, \alpha, \beta \delta) = \chi_\alpha(\beta^{-1})G(\F_p, \alpha, \delta)$
\end{claim}
\textit{Proof:} 
\[
\begin{aligned}
    G(\F_{p^r}, \alpha, \beta \delta)
      &= \sum_{x \in \F_{p^r}} \chi_\alpha(x) \omega_p^{\text{Tr}(\beta \delta x)}
      \xrightarrow{x' = \beta x}
      \sum_{x' \in \F_{p^r}} \chi_\alpha(\beta^{-1}x') \omega_p^{\text{Tr}(\delta x')}
      \\ &=
      \chi_\alpha(\beta^{-1})\sum_{x' \in \F_{p^r}} \chi_\alpha(x') \omega_p^{\text{Tr}(\delta x')}
      =
      \chi_\alpha(\beta^{-1})G(\F_{p^r}, \alpha, \delta)
\end{aligned}
\]

A direct consequence is that Gauss sums with different additive characters are proportional.

\begin{equation} \label{eq:gauss_sum_propto_one}
    G(\F_{p^r}, \alpha, \beta)
    = G(\F_{p^r}, \alpha, \beta \cdot 1)
    = \chi_\alpha(\beta^{-1})G(\F_{p^r}, \alpha, 1)
\end{equation}
Hence, they only differ by a phase, $\chi_\alpha (\beta^{-1})$.

\subsection{Quadratic Gauss Sums}
\begin{definition}[quadratic character] \label{def:quadratic_character}
    The quadratic character mod $p$ is defined as\footnote{Can also be written as Legendre symbol}
    \begin{equation} \label{eq:quadratic_character}
        \chi(a) =
        \begin{cases}
            1 & \text{for $a \ne 0$,\quad $a = u^2$,\quad $u \in \F_p$} \\
            -1 & \text{for $a \ne 0$,\quad $a \ne u^2$,\quad $u \in \F_p$} \\
            0 & \text{for $a = 0$}
        \end{cases}
    \end{equation}
\end{definition}
\begin{definition}[quadratic Gauss sums] \label{def:quadratic_gauss_sum}
    Let $p$ be an odd prime, and $a \in \F_p$, then the quadratic Gauss sum $g(a; p)$ is defined as
    \[
        g(a; p) = \sum_{x \in \F_p} \omega_p^{a x^2}
    \]
\end{definition}
\begin{claim}
    Let $\chi$ be the quadratic character mod $p$ as defined in \autoref{eq:quadratic_character}, then
    \begin{equation}\label{eq:quadratic_gauss_sum_a_to_one}
    g(a; p) = \chi(a) g(1; p)    
    \end{equation}
\end{claim}
\begin{proof}
    Expanding $g(a; p)$ in the Fourier basis yields
    \[
        \tilde{g}(k;p)
        = \frac{1}{p}\sum_{a \in \F_p} \omega_p^{-ak} g(a;p)
        = \frac{1}{p}\sum_{a \in \F_p} \sum_{x \in \F_p} \omega_p^{a(x^2-k)}
        = \frac{1}{p}\sum_{x \in \F_p} \underbrace{\sum_{a \in \F_p} \omega_p^{a(x^2-k)}}_{p\delta_{k,x^2}}
        = \begin{cases}
            2 & \text{if $k$ is a residue} \\
            0 & \text{if $k$ is a nonresidue}
        \end{cases}
    \]
    This can be written more compactly as
    \[
        \tilde{g}(k; p) = 1 + \chi(k)
    \]
    Now writing $g(a;p)$ in terms of its Fourier components yields
    \[
        g(a;p)
        = \sum_{k \in \F_p} \tilde{g}(k;p) \omega_p^{ak}
        = \sum_{k \in \F_p} (1 + \chi(k)) \omega_p^{ak}
        = 
        \underbrace{\sum_{k \in \F_p} \omega_p^{ak}}_{p\delta_{a,0}}
        +
        \underbrace{\sum_{k \in \F_p} \chi(k) \omega_p^{ak}}_{G(\chi, a)}
        = 
        p\delta_{a,0}
        +
        G(\chi, a)
    \]
    For $a=0$ we get $g(0;p) = p$, and for $a \in \F_p^\times$ we get that $g(a;p)$ is a Gauss sum $g(a;p) = G(\chi, a)$ as defined in \autoref{def:gauss_sums_over_finiite_fields}. Then using \autoref{eq:gauss_sum_propto_one} we know that $G(\chi, a) = \chi(a^{-1}) G(\chi, 1) = \chi(a^{-1}) g(1; p)$. Recalling that the inverse of a residue is a residue and that the inverse of a nonresidue is a nonresidue, and hence $\chi(a^{-1}) = \chi(a)$, yields the desired result\footnote{You can also easily see it since $\chi(a)\chi(a^{-1}) = \chi(a \cdot a^{-1})= \chi(1) = 1$, and thus $\chi(a)$ and $\chi(a^{-1})$ share the same sign}.
\end{proof}

\begin{claim} The quadratic Gauss sum $g(1; p)$ is
    \begin{equation} \label{eq:quadratic_gauss_sum_one}
        g(1; p) = \begin{cases}
            \sqrt{p}  & \text{if $p \equiv 1 \mod{4}$} \\
            i\sqrt{p} & \text{if $p \equiv 3 \mod{4}$} \\
        \end{cases}
    \end{equation}
\end{claim}
\begin{proof}
    See theorem 1.1. of \cite{murty2017evaluation}.
\end{proof}
It follows from \autoref{eq:quadratic_gauss_sum_a_to_one} and \autoref{eq:quadratic_gauss_sum_one} that
\begin{equation} \label{eq:quadratic_gauss_sum_a}
        g(a; p) = \begin{cases}
            \chi(a)\sqrt{p}  & \text{if $p \equiv 1 \mod{4}$} \\
            i\chi(a)\sqrt{p} & \text{if $p \equiv 3 \mod{4}$} \\
        \end{cases}
\end{equation}
To avoid case distinctions, define\footnote{We could have also written it more compactly $i_p = i^{\frac{1 + \chi(-1)}{2}}$}
\[
    i_p = \begin{cases}
            1  & \text{if $p \equiv 1 \mod{4}$} \\
            i & \text{if $p \equiv 3 \mod{4}$} \\
    \end{cases}
\]

\begin{claim}[multidimensional quadratic Gauss sum] \label{claim:multidim_quadratic_gauss_sum}
    Let $p$ be a prime, $n \in \mathbb{N}$, let $A \in \F_p^{n \times n}$ be a symmetric invertible matrix, then
    \[
        \sum_{\x \in \F_p^n} \omega_p^{\x^T A \x} = i_p^n p^{n/2} \chi(\det A)
    \]
\end{claim}
\begin{proof}
Since $A$ is symmetric and invertible, it can be diagonalized by an orthogonal matrix. Let us denote $A = O^T D O$, and then the sum becomes
\[
    \sum_{\x \in \F_p^n} \omega_p^{\x^T O^T D O \x}
    \xrightarrow{\x \to O\x}
    \sum_{\x \in \F_p^n} \omega_p^{\x^T D \x}
    =
    \prod_{i=1}^n \sum_{x_i \in \F_p} \omega_p^{D_i x_i^2}
\]
we recognize the sum as a quadratic Gauss sum, $g(D_i; p) = \chi(D_i) g(1; p)$, thus we get
\[
(g(1; p))^n\prod_{i=1}^n \chi(D_i)
\]
Using \autoref{eq:quadratic_gauss_sum_one}, recalling that $D_i$ correspond the eigenvalues of $A$, recalling that the determinant is the product of eigenvalues ($\det{A} = \prod_i \lambda_i^A$) and that $\chi$ is a homomorphism and hence $\chi(a \cdot b) = \chi(a) \cdot \chi(b)$, yields the desired result
\end{proof}

\begin{claim}[general quadratic sum] \label{claim:gen_quad_sum}
    Let $a, b, c \in \F_p$ be elements in odd prime field, then
    \[
    \sum_{x \in \F_p}
    \omega_p^{ax^2 + bx + c}
    =
    \begin{cases}
        p\omega_p^c \delta_{b} & \text{if $a=0$} \\
        \omega_p^{c - \frac{1}{4}a^{-1} b^2} \chi(a)g(1; p) & \text{otherwise}
    \end{cases}
    \]
    Equivalently, one may write
    \[
    \sum_{x \in \F_p}
    \omega_p^{ax^2 + bx + c}
    =
    p\omega_p^c \delta_{b} \delta_a
    +
    \omega_p^{c - \frac{1}{4}a^{-1} b^2} \chi(a)g(1; p)
    \]
    The second term is not ill-defined due to the $a^{-1}$ because for $a=0$, then $\chi(a) = 0$, and it vanishes.
\end{claim}
\begin{proof}
    For $a=0$ then the sum is just
    \[
    \omega_p^{c}
    \sum_{x \in \F_p}
    \omega_p^{bx}
    =
    p \omega_p^c \delta_b
    \]
    For $a \ne 0$ then we can convert the sum into a quadratic Gauss sum
    \[
    \begin{aligned}
    \sum_{x \in \F_p}
    \omega_p^{ax^2 + bx + c}
    &=
    \omega_p^c
    \sum_{x \in \F_p}
    \omega_p^{(ax + b)x}
    \xeq{x \to a^{-1} x}  
    \omega_p^c
    \sum_{x \in \F_p}
    \omega_p^{(x + b) a^{-1} x}
    \xeq{x \to x - \frac{1}{2}b}  
    \omega_p^c
    \sum_{x \in \F_p}
    \omega_p^{(x + \frac{1}{2}b) a^{-1} (x - \frac{1}{2}b)}
    \\&=
    \omega_p^c
    \sum_{x \in \F_p}
    \omega_p^{a^{-1}(x^2 - \frac{1}{4}b^2)}
    =
    \omega_p^{c - \frac{1}{4}a^{-1} b^2}
    \underbrace{\sum_{x \in \F_p}
    \omega_p^{a^{-1}x^2}}_{g(a^{-1};p)}
    =
    \omega_p^{c - \frac{1}{4}a^{-1} b^2}
    g(a^{-1}; p)
    \\&=
    \omega_p^{c - \frac{1}{4}a^{-1} b^2}
    \chi(a^{-1})g(1; p)
    =
    \omega_p^{c - \frac{1}{4}a^{-1} b^2}
    \chi(a)g(1; p)
    \end{aligned}
    \]
\end{proof}

\printbibliography
\end{document}